%% file: article.tex
\documentclass[12pt]{article}

\usepackage[english]{babel}
\usepackage[T1]{fontenc}
\usepackage[utf8]{inputenc}
\usepackage{graphicx}
\usepackage{amsmath,amsfonts,amssymb,amsbsy,amsthm, mathabx}
\usepackage[a4paper]{geometry}
\usepackage{tikz,pgfplots,subfig}
\usepackage{bm}
\usepackage{comment}

\pgfplotsset{compat=1.10}
\makeatletter
\newif\if@restonecol
\makeatother

\usepackage[algo2e,ruled,lined,titlenumbered,commentsnumbered]{algorithm2e}

\theoremstyle{remark}
\newtheorem*{rem}{Remark} 

\theoremstyle{plain}

\newtheorem{prop}{Property}

\newcommand{\un}[1]{\ensuremath{\underline{#1}}}
\newcommand{\uu}[1]{\ensuremath{\un{\un{#1}}}}
\newcommand{\dime}{d}

\newcommand{\sig}{\ensuremath{\uu{\sigma}}}
\newcommand{\sigadj}{\ensuremath{\uu{\widetilde{\sigma}}}}
\newcommand{\sigH}{\ensuremath{\sig_H}}

\newcommand\strain[1]{\uu{\varepsilon}\left(#1\right)}
\newcommand{\dep}{\ensuremath{\un{u}}}
\newcommand{\depadj}{\ensuremath{\un{\widetilde{u}}}}

\newcommand{\depv}{\ensuremath{\un{v}}}
\newcommand{\depH}{\ensuremath{\dep_H}}

\newcommand{\depvH}{\ensuremath{\depv_H}}

\newcommand{\efdep}{\ensuremath{\mathbf{u}}}

\newcommand{\lam}{\ensuremath{\boldsymbol{\lambda}}}
\newcommand{\Lam}{\ensuremath{\boldsymbol{\Lambda}}}

\newcommand{\eft}{\ensuremath{\mathbf{t}}}
\newcommand{\kerK}{\ensuremath{\mathbf{R}}}

\newcommand\shapef{\varphi}
\newcommand\shapev{\un{\boldsymbol{\shapef}}_H}
\newcommand\shapeC{\un{\boldsymbol{\shapef}}_C}
\newcommand\shapeF{\un{\boldsymbol{\shapef}}_F}

\newcommand{\stiff}{\ensuremath{\mathbf{K}}}
\newcommand{\schur}{\ensuremath{\mathbf{S}}}
\newcommand{\force}{\ensuremath{\mathbf{f}}}
\newcommand{\Force}{\ensuremath{\mathbf{F}}}

\newcommand{\mass}{\ensuremath{\mathbf{G}}}

\newcommand{\ma}{\ensuremath{\mathbf{M}}}
\newcommand{\na}{\ensuremath{\mathbf{N}}}
\newcommand{\pa}{\ensuremath{\mathbf{A}}}
\newcommand{\da}{\ensuremath{\mathbf{B}}}
\newcommand{\tpa}{\ensuremath{\mathbf{\tilde{A}}}}
\newcommand{\tda}{\ensuremath{\mathbf{\tilde{B}}}}
\newcommand{\efDep}{\ensuremath{\mathbf{U}}}
\newcommand{\res}{\ensuremath{\mathbf{r}}}
\newcommand{\bz}{\ensuremath{\mathbf{z}}}

\newcommand{\RRes}{\ensuremath{\mathbf{R}}}
\newcommand{\bZ}{\ensuremath{\mathbf{Z}}}
\newcommand{\bW}{\ensuremath{\mathbf{W}}}
\newcommand{\bP}{\ensuremath{\mathbf{Q}}}

\newcommand{\tLam}{\ensuremath{\boldsymbol{\delta{\Lambda}}}}
\newcommand{\tdep}{\ensuremath{\mathbf{\delta{u}}}}
\newcommand{\tDep}{\ensuremath{\mathbf{\delta{U}}}}

\newcommand{\pDep}{\ensuremath{\delta{\mathbf{U}}}}

\newcommand{\pLam}{\ensuremath{\boldsymbol{\delta{\Lambda}}}}

\newcommand{\matid}{\ensuremath{\mathbf{{I}}}}

\newcommand{\proj}{\ensuremath{\mathbf{{P}}}}

\newcommand{\domain}{\ensuremath{\Omega}}
\newcommand{\s}{\ensuremath{^{(s)}}}

\newcommand{\sT}{\ensuremath{^{(s)^T}}}

\newcommand\trace{\operatorname{tr}}

\newcommand\hooke{\mathbb{H}}
\newcommand\KA{\ensuremath{\mathrm{KA}}}
\newcommand\KAo{\ensuremath{\KA^0}}
\newcommand\KAt{\widetilde{\ensuremath{\KA}}}
\newcommand\KAH{\ensuremath{\KA_H}}

\newcommand\SA{\ensuremath{\mathrm{SA}}}
\newcommand\SAt{\ensuremath{\widetilde{\SA}}}

\newcommand{\broken}{\ensuremath{\KA(\bigcup\domain\s)}}
\newcommand{\brokeno}{\ensuremath{\KAo(\bigcup\domain\s)}}

\newcommand\ecr[2]{\ensuremath{{e_{CR_{#2}}(#1)}}}
\newcommand\ecrc[2]{\ensuremath{{e^2_{CR_{#2}}(#1)}}}
\newcommand\enernorm[2]{\|#1\|_{\hooke^{-1},#2}}

\title{Strict bounding of quantities of interest in computations based on domain decomposition\footnote{
NOTICE: this is the author’s version of a work that was accepted for publication in CMAME. Changes resulting from the publishing process, such as peer review, editing, corrections, structural formatting, and other quality control mechanisms may not be reflected in this document. Changes may have been made to this work since it was submitted for publication. A definitive version was subsequently published in Computer methods in Applied Mechanics and Engineering, 2015, doi: 10.1016/j.cma.2015.01.009.}}
\author{Valentine Rey$^1$, Pierre Gosselet$^1$, Christian Rey$^{1,2}$\\
  $(1)$ LMT-Cachan / ENS Cachan, CNRS, Université Paris-Saclay  \\
  61, avenue du pr\'esident Wilson, 94235 Cachan, France\\
  \texttt{valentine.rey@lmt.ens-cachan.fr}\\ \texttt{gosselet@lmt.ens-cachan.fr}\\ \texttt{christian.rey@lmt.ens-cachan.fr}\\
  $(2)$  SAFRAN Tech, \\ 1, rue Geneviève Aubé, 78772 Magny les Hameaux, France\\
  \texttt{christian.rey@safran.fr}
}

\begin{document}
\maketitle
\begin{abstract}
This paper deals with bounding the error on the estimation of quantities of interest obtained by finite element and domain decomposition methods. {The proposed bounds are written in order to separate the two errors involved in the resolution of reference and adjoint problems : on the one hand the discretization error due to the finite element method and on the other hand the algebraic error due to the use of the iterative solver}. Beside practical considerations on the parallel computation of the bounds, it is shown that the interface conformity can be slightly relaxed so that local enrichment or refinement are possible in the subdomains bearing singularities or quantities of interest which simplifies the improvement of the estimation. Academic assessments are given on 2D static linear mechanic problems.

\end{abstract}
\medskip

\section{Introduction}
\input{Introduction}

\section{Error estimation on quantity of interest in substructured context}
\label{section_1}
\input{DD_local_error_estimation}

\section{Improvement of the bound by local refinement}\label{section_2}
\input{Interfaces_incompatibles}

\section{Numerical experiments}\label{section_3}

\input{Adaptive_strategy}

\section{Conclusion}\label{section_4}

\input{Conclusion}

\bibliography{Biblio}
\end{document}

%% file: Introduction.tex
In order to certify structures by computations, two fundamental abilities must be gathered: (i) conducting large computations which can be achieved by using non-overlapping domain decomposition methods to exploit parallel hardware \cite{Gos06,Let94,Far94}, (ii) warrantying the quality of the results which implies to use verification techniques. The computed approximation is then completed by an estimation of its distance to the exact solution, that distance being either on a global norm \cite{Bab82,Lad86,Lad01bis,Lou03bis} or on a selection of engineering quantities of interest \cite{Pru99,Ohn01,Bec01}.
\medskip

This paper is the sequel of contributions by the authors in order to associate domain decomposition methods (DDM) and verification techniques. In \cite{Par10} it was shown how admissible fields (aka balanced residuals) could be naturally computed in parallel when solving one finite element system with classical DDM solvers like FETI or BDD. This tool was necessary to obtain guaranteed error bounds. Moreover the proposed method encapsulated classical sequential techniques \cite{Lad97,Ple11,Par06,Rey14bis} that could be employed independently on the subdomains. In \cite{Rey14} a global upper bound of the error was separated in two: the first contribution was purely due to the iterative solver employed in DDM and it could be made as small as wanted by doing more iterations, whereas the second contribution depended on the discretization and it was quasi-constant during the iterations. This result enabled us to stop iterations 
when the overall quality of the approximation could not be improved, which was much earlier than classical criteria on the norm of the residual would have suggested.

The purpose of this paper is to extend these works to the handling of quantities of interest. First we show how estimation on quantities of interest can be conducted easily within DDM with a separation of the contribution of the solver and of the discretization. Second we demonstrate that all the wanted properties can be preserved with nested kinematics at the interface which simplifies the improvement of the estimation by local adaptation.
\medskip

The question of separating the contribution of the solver and of the discretization, in goal-oriented error estimation, was addressed in several other studies with different approaches. 
In \cite{Par02,Par01}, the quantity of interest was the mean value of the unknown field in a region of the structure, the coarse mesh of the structure was seen as a decomposition of the true refined mesh which allowed to take advantage of the properties of FETI algorithm. In \cite{Pat01}, the authors proposed constant-free asymptotic upper and lower bounds but they did not connect the iteration and discretization contributions. In \cite{Mei09}, error indicators separating contributions were developed in the context of multigrids methods for Poisson and Stockes equations. They lead to a gain in terms of CPU time but the bounds were not guaranteed. Using the dual weighted residual method \cite{Bec01} for goal-oriented error estimation, a balance of discretization and iteration errors was proposed for eigenvalue problems in \cite{Ran10}. Finally, in \cite{Tav13}, the authors proposed goal-oriented a posteriori error estimator for problems discretized with mixed finite elements and multiscale domain 
decomposition that 
separated the contributions of the discretization of the subdomains from the contribution of the mortar interface. 
\medskip

In this paper, we employ classical extractor techniques \cite{Par97,Str00,Ohn01}, which leads to the definition of an adjoint problem which is solved at the same time as the forward problem using a block Krylov algorithm \cite{Saad97}. {Thanks to \cite{Par10,Rey14} we compute guaranteed upper and lower bounds of the quantity of interest at each iteration with the separation of the two sources of error for each problem (forward and adjoint), which enables us to drive the iterative solver by the error on the quantity of interest.} 

Moreover, since we consider quantities of interest defined on small supports, the adjoint problems have very localized solution. A simple strategy to improve the quality of the estimation is then to better the resolution of the adjoint problem either with a locally refined mesh \cite{Cha08} or using the partition of unity method and handbook techniques \cite{Cha09}. We show that within the framework of domain decomposition methods, it is possible to enrich or refine the kinematic of selected subdomains without impairing the exactness of the bound nor the ease of computation even if a certain incompatibility appears at the interface.

The paper is organized as follows. In Section~\ref{section_1}, we present the goal-oriented error estimation within domain decomposition framework; we propose bounds on quantities of interest that separate the contributions of the discretization and of the solver. In Section~\ref{section_2}, we show that nested interfaces can be employed which simplifies the improvement of the quality of the bounds by local refinement. Assessments are presented in Section~\ref{section_3} and Section~\ref{section_4} concludes the paper.

%% file: DD_local_error_estimation.tex
\subsection{Substructured formulation of the forward and adjoint problems}
\subsubsection{Reference problems}
Let $\mathbb{R}^\dime$ represent the physical space
($d=2$ or 3).
Let us consider the static equilibrium  of a (polyhedral) structure which occupies the open domain $\domain\subset\mathbb{R}^\dime$ and which is subjected to  given body force $\un{f}$ within $\Omega$,  to given
traction force $\un{g}$ on $\partial_g\Omega$  and to given displacement field $\dep_d$ on the  complementary part of the boundary $\partial_u\Omega$ ($\operatorname{meas}(\partial_u\Omega)\neq 0$). We  assume that the structure undergoes  small   perturbations  and  that  the  material   is  linear  elastic, characterized by Hooke's  elasticity tensor $\hooke$.  Let $\dep$ be  the unknown displacement field, $\strain{\dep}$ the symmetric part  of the gradient of $\dep$, $\sig$ the Cauchy stress tensor. In order to evaluate (localized) quantities of interest, we consider an adjoint problem loaded by a (linear) extractor $\tilde{l}$ with small support. Fields related to the adjoint problem will always be noted with a tilde (for instance adjoint displacement is $\tilde{\dep}$).

Let $\omega\subset\Omega$ be an open subset of $\Omega$. We introduce three affine subspaces, one linear space and one positive form:
\begin{itemize}
\item Affine subspace of kinematic admissible fields (KA-fields)
\begin{equation}\label{eq:KA}
  \KA(\omega)=\left\{ \dep\in \left(\mathtt{H}^1(\omega)\right)^\dime,\ \dep = \dep_d \text{ on }\partial\omega\bigcap\partial_u\domain \right\}
\end{equation}
\item Associated linear subspace $\KAo$:
\begin{equation}\label{eq:KAo}
  \KAt(\omega)=\KAo(\omega)=\left\{ \dep\in \left(\mathtt{H}^1(\omega)\right)^\dime,\ \dep = 0 \text{ on }\partial\omega\setminus\partial_g\domain \right\}
\end{equation}
\item Affine subspace of statically admissible fields (SA-fields) for the forward problem:
\begin{multline}\label{eq:SA}
  \SA(\omega)
  =\Bigg\lbrace   \uu{\tau}\in  \left(\mathtt{L}^2(\omega)\right)^{\dime\times \dime}_{\text{sym}}; \
    \forall  \depv \in  \KAo(\omega),\ \\ \int\limits_\omega
  \uu{\tau}:\strain{\depv}    d\domain    =    \int\limits_\omega   \un{f} \cdot\depv    d\domain +
  \int\limits_{\partial_g\domain\bigcap\partial\omega} \un{g}\cdot\depv dS = l(\depv)  \Bigg\rbrace
\end{multline}
We note $l$ the linear form which gathers the loads of the forward problem.
\item Affine subspace of statically admissible fields (SA-fields) for the adjoint problem:
\begin{multline}\label{eq:SA_adjoint}
  \SAt(\omega)
  =\Bigg\lbrace   \uu{\tau}\in  \left(\mathtt{L}^2(\omega)\right)^{\dime\times \dime}_{\text{sym}}; \
    \forall  \depv \in  \KAo(\omega),\ \int\limits_\omega
  \uu{\tau}:\strain{\depv}    d\omega   =    \tilde{l} (\depv)   \Bigg\rbrace
\end{multline}
\item Error in constitutive equation 
\begin{equation}\label{eq:ecr}
  \ecr{\dep,\sig}{\omega}= \enernorm{\sig-\hooke:\strain{\dep}}{\omega}
\end{equation}
where ${\enernorm{\uu{x}}{\omega}}=\displaystyle \sqrt{\int_\omega \left( \uu{x}: {\hooke}^{-1}:\uu{x} \right)d\omega}$
\end{itemize}

The forward and adjoint problem set on $\domain$ can be formulated as:
\begin{equation}\label{eq:refpb}
  \text{Find } \left(\dep_{ex},\sig_{ex}\right)\in\KA(\domain)\times\SA(\domain) \text{ such  that } \ecr{\dep_{ex},\sig_{ex}}{\domain}=0
\end{equation}
\begin{equation}\label{eq:refpbad}
  \text{Find } \left(\depadj_{ex},\sigadj_{ex} \right)\in\KAo(\domain)\times\SAt(\domain) \text{ such  that } \ecr{\depadj_{ex},\sigadj_{ex}}{\domain}=0
\end{equation}
The solution to these problems, named ``exact'' solutions, exist and are unique.

\subsubsection{Substructured formulation}
 Let us consider a decomposition of domain $\Omega$ in $N_{sd}$ open subsets $\Omega^{(s)}$ such that $\Omega^{(s)}\bigcap \Omega^{(s')}=\emptyset$ for $s\neq s'$ and $\bar{\Omega}=\bigcup_s \bar{\Omega}^{(s)}$. The interface between subdomains $\Gamma^{(s,s')}=\bar{\Omega}^{(s)}\bigcap \bar{\Omega}^{(s')}$ is supposed to be regular enough for traces 
of locally admissible fields to be well defined. $\un{n}$ is the outer normal vector. 

The mechanical problem on the substructured configuration writes :
\begin{equation}
  \forall s \left\{
   \begin{aligned}
&\dep^{(s)} \in \KA(\Omega^{(s)})  \\
& \sig^{(s)} \in \SA(\Omega^{(s)}) \\
&  \ecr{\dep^{(s)},\sig^{(s)}}{\Omega^{(s)}}=0
\end{aligned}  \right. \text{ and } \forall (s,s')\left\{
    \begin{aligned}
   & \trace(\dep^{(s)})=\trace(\dep^{(s')}) \text{ on } \Gamma^{(s,s')}\\
 &  \sig^{(s)} \cdot\underline{n}^{(s)} +\sig^{(s')}\cdot \underline{n}^{(s')} =\underline{0} \text{ on } \Gamma^{(s,s')}
    \end{aligned} 
  \right.
\end{equation}
We assume the same domain decomposition is used for the forward and adjoint problems. The adjoint problem satisfies the same type of substructured formulation:
the fields $\dep$ (respectively $\depadj$) and $\sig$ (resp. $\sigadj$) are globally admissible if $(\dep^{(s)},\sig^{(s)})$ (resp. $(\depadj^{(s)},\sigadj^{(s)})$) are kinematically and statically admissible inside each $\Omega^{(s)}$ and if the displacements are continuous and tractions balanced on interfaces.
The set of fields $\dep$ (or $\depadj$) defined on $\domain$ such that they are kinematically admissible inside each domain $\Omega^{(s)}$ without interface continuity is a broken space which we note \broken\ and \brokeno. Similarly, we denote $\SA(\bigcup\Omega\s)$ the broken space of statically admissible fields inside subdomains.

\begin{rem}
The error in constitutive relation is by essence a parallel quantity: 
\begin{equation}
\forall \dep\in\broken, \forall \sig\in\SA(\bigcup\Omega\s),\quad \ecr{\dep,\sig}{\Omega}=\sqrt{\sum_s\ecrc{\dep\s,{\sig}\s}{\Omega\s}}
\end{equation}
\end{rem}

\subsection{Finite element approximation for the forward and adjoint problems}

Both forward and adjoint problems are solved with a finite element approximation. Note that the mesh used to solve the adjoint problem can be completely different from the mesh used for the forward problem. We detail the finite element approximation for the forward problem (the principle is the same for the adjoint problem).
\subsubsection{Global form}
We associate with the mesh of $\Omega$ the finite-element subspace $\KAH(\domain)$. The finite element approximation consists in searching:
\begin{equation}
  \begin{aligned}
    \depH \in\KAH(\domain),&\qquad  \sigH=\hooke:\strain{\depH}   \\
    \int_{\domain}
  \sigH:\strain{\depvH}   d\domain  &=   \int_{\domain} \un{f}\cdot\depvH   d\domain   +
  \int_{\partial_g\domain} \un{g}\cdot\depvH dS = l(\depvH),\qquad \forall\depvH\in \KAo_H(\domain)
    \end{aligned}
\end{equation}
We introduce the matrix $\shapev$ of  shape functions which form a basis of $\KAH(\domain)$ 
and the  vector of  nodal unknowns  $\efdep$. Therefore, we have $\depH=\shapev \efdep$ and obtain the following linear system:
\begin{equation}\label{eq:globalFE}
\stiff \efdep = \force
\end{equation}
where $\stiff$ is  the (symmetric positive definite) stiffness  matrix and $\force$ is the vector of generalized forces (which includes imposed displacement).

\subsubsection{Substructured form}
For both problems, we assume that the mesh of  ${\domain}$ and the substructuring are conforming. This hypothesis implies that each element only belongs to one subdomain and nodes are matching on the interfaces. Each degree of freedom is either located inside a subdomain (Subscript i) or on its boundary (Subscript b).

Let  $\eft^{(s)}$ be the discrete trace operator, so that $\efdep_b\s=\eft^{(s)}\efdep\s$. Let us introduce the unknown nodal reaction on the interface $\lam\s$, the equilibrium of each subdomain writes:
\begin{equation}\label{eq:efddeq}
\stiff\s \efdep\s = \force\s + {\eft\s}^T \lam\s
\end{equation}
Let $(\pa\s)$ and $(\da\s)$ be the primal and dual assembly operator so that the discrete counterpart of the interface admissibility equations is:
\begin{equation}\label{eq:efintadmiss}
\left\{\begin{aligned}
\sum_s\da\s \eft\s\efdep\s &=0 \\
\sum_s\pa\s \lam\s &=0 \\
\end{aligned}\right.
\end{equation}

Note that in the case when Subdomain $s$ has not enough Dirichlet boundary conditions then the  reaction has to balance the subdomain with respect to rigid body motions. Let $\kerK\s$ be a basis of $\operatorname{ker}(\stiff\s)$, we have:
\begin{equation}\label{eq:ortho_modes_rigides}
{\kerK\s}^T(\force\s + {\eft\s}^T \lam\s) =0
\end{equation}
 
\subsubsection{Domain decompostion solvers}
First, let us briefly present the two classical domain decomposition solvers BDD and FETI (see Algorithm~\ref{alg:bdd_incomp} with for now $\mathbf{M}=\mathbf{N}=\mathbf{I}$ for BDD and \cite{Rey14} for FETI, see also \cite{Gos06} and its bibliography for more details). In BDD solver, a unique interface displacement is introduced so that the continuity of displacement is automatically verified. Using the primal Schur complement, the problem can be rewritten with interface quantities. In FETI solver, a unique set of nodal forces on the interface is introduced so that the balance of forces is verified. The dual Schur complement allows to write the problem with interface quantities.

Then, the resolution of interface problem is based on iterative Krylov solver, namely a preconditioned conjugate gradient. The classical preconditioners make use of scaled assembling operators $(\tpa\s)$ and $(\tda\s)$ solutions of equations \cite{Rix99bis,Kla01}:
\begin{equation}
\sum_s \pa\s \tpa\sT = \matid \qquad \text{and}\qquad \sum_s \da\s \tda\sT = \matid 
\end{equation}
During the algorithm, two types of mechanical problems are solved locally (on each subdomain): problems with Dirichlet conditions on the interface (subscript D) and problems with Neumann conditions on the interface (subscript N). They correspond to the following procedures: 
\smallskip

\vline
\begin{minipage}[t]{.45\textwidth}
$(\lam\s_D,\efdep_D\s)=\mathtt{Solve}_D(\efdep\s_b,\force\s)$:
\begin{equation*}
\left\{ \begin{aligned}
&\stiff\s \efdep_D\s = \force\s + {\eft\s}^T \lam_D\s \\
&\eft\s\efdep_D\s=\efdep\s_b
\end{aligned}
\right.
\end{equation*}
\end{minipage}\hfill\vline
\begin{minipage}[t]{.45\textwidth}
$(\efdep_N\s)=\mathtt{Solve}_N(\lam_N\s,\force\s)$
\begin{equation*}
\left\{ \begin{aligned}
&\stiff\s \efdep_N\s = \force\s + {\eft\s}^T \lam_N\s\\
&\text{where }(\lam_N\s)_s \text{ satisfy Eq.~} \eqref{eq:ortho_modes_rigides}
\end{aligned}
\right.
\end{equation*}

\end{minipage}

\medskip

When developing these methods, we get:
\begin{equation*}
\begin{aligned}
&(\efdep_D\s)_i = {\stiff\s_{ii}}^{-1}\left( \force_i\s - \stiff\s_{ib} \efdep_b\s\right)\qquad\text{and}\qquad (\efdep_D\s)_b= \efdep_b\s\\
&\lam\s_D = \schur\s \efdep_b\s -\force\s_b+\stiff\s_{bi} {\stiff\s_{ii}}^{-1}\force\s_i\\
&\efdep_N\s = {\stiff\s}^+\left( \force\s + {\eft\s}^T \lam_N\s\right)
\end{aligned}
\end{equation*}
where $\schur\s=\left(\stiff\s_{bb}-\stiff\s_{bi} {\stiff\s_{ii}}^{-1}\stiff\s_{ib}\right)$ is the well-known Schur complement matrix, and ${\stiff\s}^+$ is a pseudo-inverse of ${\stiff\s}$. Thanks to projector ($\proj_1$) and initialization (method $\mathtt{Initialize}$), the Neumann problems are well-posed (see \cite{Gos06}). \medskip

If the forward and adjoint problems are approximated on the same mesh, they share the same stiffness matrix and they can be solved simultaneously using a block-conjugate gradient. In Algorithm~\ref{alg:bdd_incomp}, we use capital letter to note blocks of forward/adjoint vectors: typically $\efDep=[\efdep,\tilde{\efdep}]$, same notation holds for the residual $\RRes=[\res,\tilde{\res}]$, preconditioned residual $\bZ=[\bz,\tilde{\bz}]$, search directions $\bW$ and $\bP$;  $\alpha$ is a $2\times 2$ matrix. One has to pay attention to the potential (yet unlikely) singularity of the $2\times 2$ matrix $(\bP^T\bW)$ which needs to be inverted. 

\subsubsection{Description of admissible fields}
At every iteration of the solver, a totally parallel computation of the following fields is possible~\cite{Par10}:
\begin{itemize}
\item $\dep_D=(\dep_D\s)_s \in \KA(\domain)$ and $\depadj_D=(\depadj_D\s)_s \in \KAo(\domain)$: globally admissible displacement fields, $(\lam_D\s,\widetilde{\lam}_D\s)$ are the associated nodal reaction which are not balanced before convergence.
\item $\dep_N=(\dep_N\s)_s  \in \broken$ and $\depadj_N=(\depadj_N\s)_s \in \brokeno$: displacements associated to a given balanced Neumann condition $\lam_N\s$ (resp. $\widetilde{\lam}_N\s$), they are not necessarily continuous across interfaces.
\item $\sig_N\s = \hooke:\strain{\dep_N\s}$: the stress field associated to $\dep_N\s$. It can be employed (with additional input $\lam_N\s$) to build stress fields $ \hat{\sig}_N\s$ which are statically admissible $\hat{\sig}_N=(\hat{\sig}_N\s)_s\in\SA(\Omega)$. Same property holds for the adjoint problem, we can build $\hat{\sigadj}_N=(\hat{\sigadj}_N\s)_s\in\SAt(\Omega)$.
\end{itemize}
As a consequence, for both forward and adjoint problems, kinematically displacements fields and statically stress fields are available even if the solver is not converged. Those fields are used to compute global error estimation for both problems. Indeed, we have the following guaranteed upper bounds which can be computed in parallel~\cite{Par10}:
\begin{equation}
\begin{aligned}
 \vvvert  \dep_{ex} -\dep_D  \vvvert_\domain&\leq\ecr{\dep_D,\hat{\sig}_N}{\Omega} \\  \vvvert  \depadj_{ex} -\depadj_D  \vvvert_\domain  &\leq\ecr{\depadj_D,\hat{\sigadj}_N}{\Omega}\end{aligned}
\end{equation} 
where $\vvvert  \un{x}  \vvvert_\domain= {\enernorm{\strain{\un{x}}}{\Omega}}$.

\subsection{Error estimation of quantities of interest}
As said earlier the adjoint problem is meant to extract one quantity of interest in the forward problem. Let $I_{ex}=\tilde{l}(\dep_{ex})$ be the unknown exact value of the quantity of interest. $\tilde{l}({\dep}_D)$ is an approximation of this quantity of interest.

\subsubsection{Upper and lower bounds with separated contributions}
 
We demonstrate two results which provide upper and lower bounds on the quantity of interest, {as a product of error estimates on the forward and adjoint problems with, for each of these factors, a separation of contributions of the solver and of the discretization.}
 
 \begin{prop}\label{prop:separation_1_DD} Using notation of Algorithm~\ref{alg:bdd_incomp}, we have
\begin{equation}
\arrowvert \mathit{I}_{ex} -  \tilde{l}(\dep_D) - \mathit{I}_{HH,1} \arrowvert \leq 
 \left( \sqrt{\res^T\bz} + \ecr{\dep_N,\hat{\sig}_N}{\Omega} \right) \left( \sqrt{\widetilde{\res}^T\widetilde{\bz}} +\ecr{\depadj_N,\hat{\sigadj}_N}{\Omega}  \right)
\end{equation}
with 
\begin{equation}\label{eq:Ihh1}
 I_{HH,1}=\int_\domain (\hat{\sig}_N-{\hooke}\strain{\dep_D}):\strain{\depadj_D}  d\domain
\end{equation}
\end{prop}
\begin{proof}
This result is a direct application of the results demonstrated in  \cite{Lad06,Lad08}:
 \begin{equation}
\arrowvert \mathit{I}_{ex} -  \tilde{l}(\dep_D) - \mathit{I}_{HH,1} \arrowvert \leq   { \ecr{\dep_D,\hat{\sig}_N}{\Omega}}   { \ecr{\depadj_D,\hat{\sigadj}_N}{\Omega}}  
\end{equation}
with $ I_{HH,1}$ given in~\eqref{eq:Ihh1}. The separation between the sources of error is obtained using the triangular inequality and Lemma 2 presented in \cite{Rey14} (which proves the equality $\vvvert{\dep_N-\dep_D}\vvvert_\domain=\sqrt{\res^T\bz}$), it slightly improves the result given in remark 25 of \cite{Rey14}:
 \begin{equation}
  \begin{aligned}
   \ecr{\dep_D,\hat{\sig}_N}{\Omega}&=\enernorm{\hat{\sig}_N-\hooke:\strain{\dep_D}}{\Omega}\\
   &=\enernorm{\hat{\sig}_N-\hooke:\strain{\dep_N} + \hooke:\strain{\dep_N - \dep_D}}{\Omega}\\
   &\leq \enernorm{\hat{\sig}_N-\hooke:\strain{\dep_N}}{\Omega} + \vvvert{\dep_N-\dep_D}\vvvert_\domain\\
   &\leq \ecr{{\dep}_N,\hat{\sig}_N}{\Omega}+ \sqrt{\res^T\bz}
  \end{aligned}
 \end{equation}
\end{proof}

{
The quantity $\sqrt{\res^T\bz}$ (respectively $\sqrt{\widetilde{\res}^T\widetilde{\bz}}$) is a measure of the residual of the forward (resp. adjoint) problem: it corresponds to a pure algebraic error. The quantity $\ecr{\dep_N,\hat{\sig}_N}{\Omega\s}$ (resp. $\ecr{\depadj_N,\hat{\sigadj}_N}{\Omega\s}  )$) is a parallel estimate of the discretization error of the forward (resp. adjoint) problem.}
{
As a consequence, we can define a new stopping criterion for the solver: the idea is to stop iterations when the algebraic error becomes negligible in comparison with the discretization error. This new stopping criterion is described in alg.~\ref{alg:ddstop}.}

\begin{algorithm2e}[h!]\caption{{DD block-solver with adapted stopping criterion}}\label{alg:ddstop}
Set $\alpha>1$, $\beta\geqslant 1$, $\widetilde{\alpha}>1$ and $\widetilde{\beta}\geqslant 1$\;
Initialize, get $(\efDep_N\s,\Lam_N\s,\RRes,\bZ)$\;%
Estimate discretization error for reference problem $e^2=\sum_s\ecrc{\dep_N\s,\hat{\sig}_N\s}{\Omega\s}$ and adjoint problem $\widetilde{e}^2=\sum_s\ecrc{\depadj_N\s,\hat{\sigadj}_N\s}{\Omega\s}$\;
\While{$\sqrt{\res^T\bz}>e/\alpha$ and $\sqrt{\widetilde{\res}^T\widetilde{\bz}}>\widetilde{e}/\widetilde{\alpha}$}{
\While{$\sqrt{\res^T\bz}>e/(\alpha \beta)$ and $\sqrt{\widetilde{\res}^T\widetilde{\bz}}>\widetilde{e}/\widetilde{\alpha}\widetilde{\beta}$}{
Make Alg.~\ref{alg:bdd_incomp} iterations, get $(\efDep_N\s,\Lam_N\s)$
}
Update error estimators $e^2=\sum_s\ecrc{\dep_N\s,\hat{\sig}_N\s}{\Omega\s}$ and $\widetilde{e}^2=\sum_s\ecrc{\depadj_N\s,\hat{\sigadj}_N\s}{\Omega\s}$ \;
}
\end{algorithm2e}

{
The objective being to have the norm of the residual $\alpha$ times smaller than the estimation of the discretization error for both problems ($\alpha=\widetilde{\alpha}=10$ is a typical value), the coefficient $\beta=\widetilde{\beta}$ (typically 2) takes into account the small variation of the estimate $e$ and $\widetilde{e}$ of the discretization error along iterations.}

A second (twice better) bound can be derived in a similar way using another result of  \cite{Lad06,Lad08}.
\begin{prop}\label{prop:separation_2_DD}
\begin{equation}
\arrowvert \mathit{I}_{ex} - \tilde{l}(\dep_D) - \mathit{I}_{HH,2} \arrowvert \leq  
\frac{1}{2}  (\sqrt{\res^T\bz} + \ecr{\dep_N,\hat{\sig}_N}{\Omega} ) ( \sqrt{\widetilde{\res}^T\widetilde{\bz}} +\ecr{\depadj_N,\hat{\sigadj}_N}{\Omega}  )
\end{equation}
with
\begin{equation}
 I_{HH,2}=\frac{1}{2}\int_\domain \left(\hat{\sig}_N-{\hooke}\strain{\dep_D}\right):{\hooke}^{-1}:\left(\hat{\sigadj}_N+{\hooke}\strain{\depadj_D}\right)  d\domain 
\end{equation}
\end{prop}

\medskip

\subsubsection{Properties of $I_{HH,1}$}
If the forward and adjoint problems share the same mesh, the quantity $I_{HH,1}$ can be rewritten in terms of quantity on the interface, which simplifies its evaluation:
\begin{prop}\label{prop:reecriture_IHH1}
\begin{equation}\label{eq:ihh1b}
 I_{HH,1}=\sum_s \lam_D\sT {\eft\s} {\widetilde{\efdep}}_D^{(s)}
\end{equation}
\end{prop}
\begin{proof}
The forward problems verified by $\dep_D \in \KA(\domain) $ and $\dep_{ex} \in \KA(\domain) $ with a testing field $\widetilde{\dep}_D \in \KAo(\domain)$ read: 
\begin{equation}\label{eq:UD_ref}
\begin{aligned}
\sum_s \int_{\Omega_s} \strain{\un{{u}}_D^{(s)}} :\hooke :\strain{\un{{\widetilde{u}}}_D^{(s)}} d\Omega &=\sum_s l\s(\un{{\widetilde{u}}}_D^{(s)})   + \sum_s \lam_D\sT {\eft\s} {\widetilde{\efdep}}_D^{(s)}
\end{aligned}
\end{equation}
\begin{equation}\label{eq:Uex_ref}
\begin{aligned}
\sum_s \int_{\Omega_s} \strain{\dep_{ex}^{(s)}} :\hooke :\strain{\un{\widetilde{u}}_D^{(s)}} d\Omega &=\sum_s l\s(\un{{\widetilde{u}}}_D^{(s)}) + \sum_s \int_{\Gamma\s} (\sig_{ex} \cdot \un{n}\s) \cdot\un{\widetilde{u}}_D^{(s)}ds 
\end{aligned}
\end{equation}
Since $\un{\widetilde{u}}_D=(\un{\widetilde{u}}_D^{(s)})_s\in \KAo(\domain)$ and $\dep_{ex}$ is the exact solution, we have:
\begin{equation}
\sum_s \int_{\Gamma\s} (\sig_{ex} \cdot \un{n}\s)\cdot \un{\widetilde{u}}_D^{(s)}ds = 0
\end{equation}
while since $\hat{\sig}_N$ is statically admissible and $(\un{\widetilde{u}}_D^{(s)})_s\in \KAo(\domain)$, we have:
\begin{equation}
  I_{HH,1}=\int_\domain (\hat{\sig}_N-{\hooke}\strain{\dep_D}):\strain{\depadj_D}  d\domain=\int_\domain ({\sig}_{ex}-{\hooke}\strain{\dep_D}):\strain{\depadj_D}  d\domain
\end{equation}
Therefore, by substracting equation \eqref{eq:UD_ref} from equation \eqref{eq:Uex_ref}, we obtain \eqref{eq:ihh1b}.
\end{proof}

When the solver has reached convergence, if the meshes of the forward and adjoint problems are identical, the term $\textit{I}_{HH,1}$ is equal to zero (as a consequence of the Galerkin orthogonality). Thanks to the new expression of $\textit{I}_{HH,1}$, we can prove the following property.
\begin{prop}\label{thm:cas_particulier_primal}
Using a BDD solver, at each iteration, $\textit{I}_{HH,1}=0$.
\end{prop}
\begin{proof}
Using the notation of BDD's algorithm (alg.~\ref{alg:bdd_incomp} with $\mathbf{M}=\mathbf{N}=\mathbf{I}$), we have:
\begin{equation}
I_{HH,1}= \res^T \widetilde{\efdep}
\end{equation}
The iterations of the block-CG can be defined by the following Galerkin scheme~\cite{Leary:BlockCG,Saa00}:
\begin{equation}\left\{
\begin{aligned}
\efDep_i&\in \efDep_0 + \mathcal{K}_i \\
\RRes_i &\perp \mathcal{K}_i
\end{aligned}\right.
\end{equation}
where $\efDep=[\efdep,\tilde{\efdep}]$, $\RRes=[\res,\tilde{\res}]$, and $\mathcal{K}_i$ is the Krylov subspace associated with the preconditioned operator and the initial block of residuals. This property implies that 
\begin{equation}
I_{HH,1}= \res^T \widetilde{\efdep} =  \res_0^T \widetilde{\efdep}_0
\end{equation}
In order to prove the nullity of this term, one has to look more in detail at BDD's initialization. We refer to \cite{Gos06} and simply recall that $\res_0\in \operatorname{range}(\proj_1^T)$ while $\widetilde{\efdep}_0\in \operatorname{ker}(\proj_1)$.
\end{proof}

\subsection{Assessment}

Let us consider a rectangular linear elastic structure $\Omega=[-3l;3l]\times[-3l;3l]$ with homogeneous Dirichlet boundary conditions.  Young modulus is 1 Pa and Poisson coefficient is $0.3$. The domain is subjected to a polynomial body force such that the exact solution is known: 
\begin{equation*}
\underline{u}_{ex} = (x+3l)(x-3l)(y+3l)(y-3l)\left((y-3l)^2 \underline{e}_x + (y+3l)\underline{e}_y \right)
\end{equation*}
The quantity of interest is the average of the normal stress in the direction $\un{e}_y$ in a small square of side $l/10$ centered at the origin.
 \begin{equation}
\tilde{l}(\dep)=\frac{1}{mes(\omega)} \int_{\omega} (\un{e}_y\otimes\un{e}_y):\strain{\dep} d\omega 
\end{equation}
The mesh is made out of first order Lagrange triangles, which leads to 2352 degrees of freedom. 
As shown in Figure \ref{fig:poutre_decoupage} the structure is decomposed into 9 regular subdomains.
 \begin{figure}[ht]
\centering
\includegraphics[width=.3\textwidth]{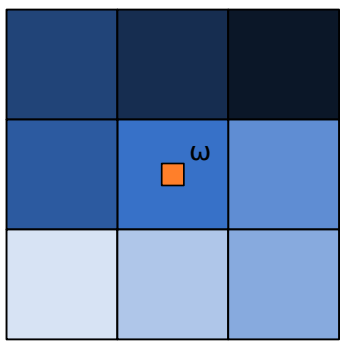}\caption{Substructuring}\label{fig:poutre_decoupage}
 \end{figure}

We solved the problem with a BDD solver (primal approach). We used the Element Equilibration Technique \cite{Lad91} to build statically admissible stress fields; each elementary problem was solved using an interpolation of degree 4.
For the sake of simplicity, we note: $e=\ecr{\dep_N,\hat{\sig}_N}{\Omega}$ and $\widetilde{e}=\ecr{\depadj_N,\hat{\sigadj}_N}{\Omega}$. The results in Table~\ref{tab:poutre_1} are given when the solver is converged (when the residual of both forward and adjoint problem are negligible with respect to the loads); they show that the error estimation in a substructered context is as efficient as in a sequential context. We recall that the slight difference is due to the fact that the static admissible fields are not identical in sequential and in domain decomposition methods.

\begin{table}[ht]\centering%
\begin{tabular}{|c||c|c|c|c|c|c|c|}
\hline 
case & $e$  & $\widetilde{e}$ &  $I_H$  & $I_{HH,2}$  & $\frac{1}{2}e\widetilde{e}$  & $I_{ex}$\\ 
\hline \ 
Sequential& 2.417& 10.23  & 2.916 & 3.305 $10^{-4}$ & 12.36  & 2.916
  \\
\hline 
Substructured&  2.418  & 10.23   &  2.916& 3.858 $10^{-4}$& 12.37&  2.916
 \\ 
\hline 
\end{tabular} 
\caption{Substructured and sequentiel goal-oriented error estimations}\label{tab:poutre_1}
\end{table}

\begin{figure}[ht]
\begin{minipage}{.49\textwidth}
\begin{tikzpicture}
\begin{semilogyaxis}[
width = 0.99\textwidth,
title=Forward problem,
xlabel=iterations,
legend style={at={(0.06,0.35)}, anchor=north west}] 
\addplot [draw=green, mark=triangle*] table[x=iterations,y=res] {poutre_1.txt};
\addlegendentry{ {\scriptsize $\sqrt{\res^T\bz}$ } }
\addplot [draw=red, mark=diamond] table[x=iterations,y=e] {poutre_1.txt};
\addlegendentry{{\scriptsize \ecr{{\dep}_N, \hat{\sig}_N}{\Omega}}}
\end{semilogyaxis}
\end{tikzpicture}
\end{minipage}
\begin{minipage}{.49\textwidth}
  \begin{tikzpicture}
\begin{semilogyaxis}[
width = 0.99\textwidth,
title=Adjoint problem,
xlabel=iterations,
legend style={at={(0.06,0.35)}, anchor=north west}] 
\addplot [draw=green, mark=triangle*] table[x=iterations,y=res_adj] {poutre_1.txt};
\addlegendentry{ {\scriptsize $\sqrt{\widetilde{\res}^T\widetilde{\bz}}$ } }
\addplot [draw=red, mark=diamond] table[x=iterations,y=e_adj] {poutre_1.txt};
\addlegendentry{{\scriptsize \ecr{{\depadj}_N, \hat{\sigadj}_N}{\Omega}}}
\end{semilogyaxis}
\end{tikzpicture}
 \end{minipage}
\caption{Error estimates with separation of sources}\label{fig:poutre_1_res}
\end{figure}
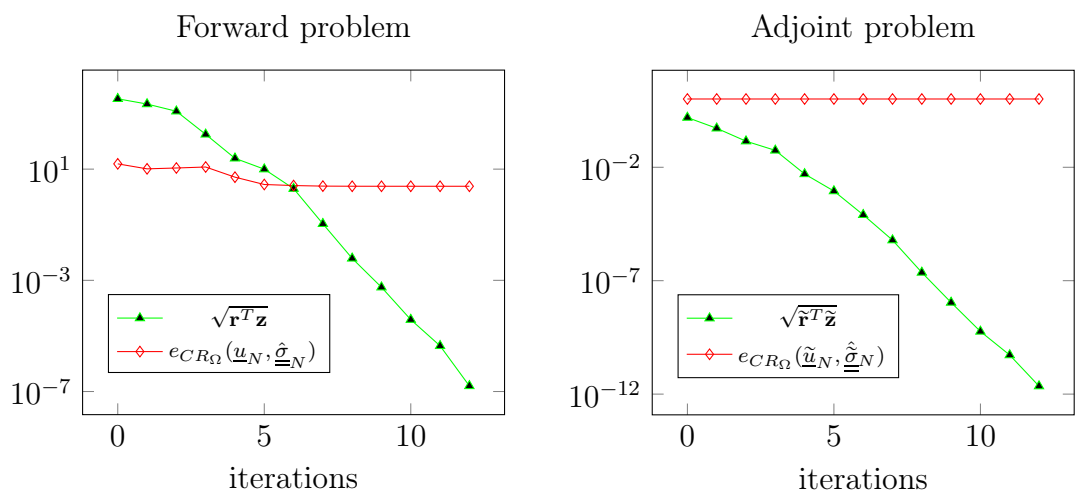

The plots on Figure~\ref{fig:poutre_1_res} show the evolution of the two contributions to the bound during the iterations of the conjugate gradient. We observe that the discretization error is almost constant whereas the algebraic error decreases along the iterations. On that example after 6 iterations, the solver errors of both forward and adjoint problems are negligible with respect to the contribution of the discretization, extra iterations are then useless, they do not improve the quality of the approximation.

%% file: Interfaces_incompatibles.tex
\subsection{Motivation}

The upper bounds obtained in Property~\ref{prop:separation_1_DD} and Property~\ref{prop:separation_2_DD} always involve the product between the error estimators of the forward and adjoint problem. Accordingly increasing the quality of either (or both) the forward or the adjoint problems necessarily improves the bounds.  For instance, since the solution of the adjoint problem is highly localized, it is often easier to better its estimation, either by mesh refinement or by the injection of a priori knowledge \cite{Gra06,Wae12,Cha08,Cha09}. 

Adapting a mesh, even for a local quantity of interest, is a global problem: simple algorithms like the one based on bisection \cite{NVB} propagate on the whole domain; parallel remeshing involves exchanges between neighbors (and iterations if good load balance is also wished) \cite{ParaRemesh}. Enrichment also implies integrating functions on parts of the domain. In this section, we try to benefit the domain decomposition in order to simplify these procedure. The idea is to improve the approximation only on the subdomains which most contribute to the estimators. We show that a class of adaptation algorithms  can be realized completely locally without impairing the exactness and the computability of the bound: algorithms which modify the boundary of the treated subdomains only by adding new degrees of freedom. This leads to what we call ``nested interfaces'' which means that the trace of the initial (coarse) kinematic is a subspace of the trace of the new (fine) kinematic.

Typically, bisections stopped at the boundary of the subdomains belongs to that class of algorithms, so do hierarchical h and p refinement and most enrichment methods. In the examples, for simplicity reasons, we use uniform h-refinement in the subdomains. Of course such procedures open the question of balancing the computational load between subdomains, which will be the subject of future work. Anyhow, a classical answer is to split the structure into many subdomains and to randomly attribute several subdomains to each computing unit with the hope that the to-be-refined subdomains will be evenly distributed.

The rest of the section is organized as follows. In Subsection~\ref{ssec:interpo} we show how a master/slave algorithm enables us to treat nested interfaces in order to preserve the ability to build  kinematic and static admissible fields. This property would be much more difficult to conserve with  mortar methods \cite{Ber94,Bel99}. The algorithm is presented in Subsection~\ref{ssec:alginc}; we choose to solve simultaneously the forward and the adjoint problems on the locally refined mesh, since there is no real-extra cost in solving both problems at the same time. In Subsection~\ref{ssec:assessinc}, we verify that the chosen strategy enables us to improve the bound on the estimation.

\subsection{Interpolation of displacements and forces}\label{ssec:interpo}
We assume that after a first resolution, some subdomains have been locally refined (either in $h$ or $p$ sense) or enriched, so that new degrees of freedom have been inserted on the boundary of the domain whereas former boundary nodes have been preserved. Figure~\ref{fig:incompatibilite_2SD} presents a $H/2$ refinement in the case of 2 subdomains. In order to keep the ability of reconstructing globally admissible fields, we decide the (coarse) interface to be the master when imposing displacements or forces on subdomains, and we keep the assembling operators unchanged. We thus only need to derive the interpolation matrix $\na\s$ which defines local (fine) displacements $\efdep_F\s$ from interface (coarse) displacement $\efdep_C\s=\pa\sT\efDep$: $\efdep_F\s=\na\s\efdep_C\s$; and the interpolation matrix $\ma\s$ which defines local (fine) nodal forces $\lam_F\s$ from interface (coarse) forces $\lam_C\s=\da\sT\Lam$: $\lam_F\s=\ma\s\lam_C\s$.

\begin{figure}[ht]
\centering
\includegraphics[width=.4\textwidth]{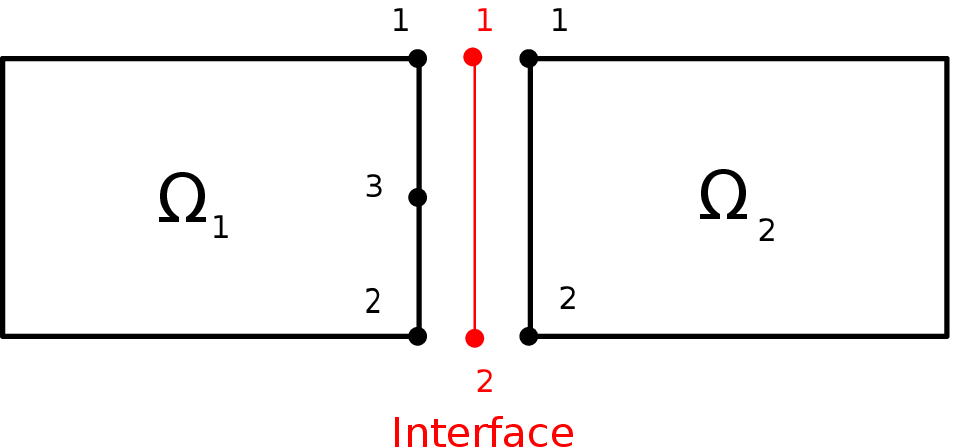}\caption{Titre}\label{fig:incompatibilite_2SD}
\end{figure}

Let $\shapeC\s$ and $\shapeF\s$ be the matrices of the coarse and fine shape functions, so that displacements write $\dep_C\s=\shapeC\s\efdep_C\s$ on the interface and  $\dep_F\s=\shapeF\efdep_F\s$ on the boundary. Let $\mathbf{X}\s_F$ be the coordinates of the fine nodes, we have $\shapeF\s(\mathbf{X}\s_F)=\matid$.
Imposing the continuity of the displacement writes:
\begin{equation*}
\begin{aligned}
&\dep_C\s=\shapeC\s\efdep_C\s= \shapeF\s\efdep_F\s=\dep_F\s \\
&\Leftrightarrow \shapeC\s(\mathbf{X}_F)\s \efdep_C\s = \efdep_F\s,
\end{aligned}
\end{equation*}
which defines the interpolation matrix for displacement:
\begin{equation}
\begin{aligned}
\na\s&= \shapeC\s(\mathbf{X}\s_F).
\end{aligned}
\end{equation}
Matrix $\na\s$ simply imposes extra fine nodes to follow the displacement imposed by the coarse discretization.\medskip

The generation of statically admissible fields requires an algorithm, noted $\mathcal{G}$ in \cite{Par10}, which defines a continuous distribution of traction $\un{t}$ from nodal reactions: $\un{t}=\mathcal{G} \lam $. 
Then transfer operator $\ma\s$ must ensure that the same continuous tractions are created from the coarse and the fine information: $\un{t}\s_F=\un{t}_C\s$. In other words $\mathcal{G}_F\ma\s = \mathcal{G}_C $.

Algorithm $\mathcal{G}$ in \cite{Par10} consists in finding a distribution of traction on the same finite element subspace as the displacement, under the constraint to preserve the mechanical work:
\begin{equation}\label{eq:def_effortCF}
\begin{aligned}
\un{t}\s_C &= \shapeC\s \left(\mass_C\s \right)^{-1} \lam\s_C &\text{with } \mass_C\s&=\int_\Gamma \shapeC\sT \shapeC\s d\Gamma\\
\un{t}\s_F &= \shapeF\s \left(\mass_F\s\right)^{-1} \lam\s_F &\text{with } \mass_F\s&=\int_\Gamma \shapeF\sT \shapeF\s d\Gamma
\end{aligned}
\end{equation}
$\mass_C\s$ and $\mass_F\s$ are the $L^2(\Gamma)$ mass matrices of the coarse and fine interpolations. 
Because of the chosen interpolation, it is sufficient to satisfy the relationship $\un{t}\s_F=\un{t}_C\s$ on the fine nodes, which leads to:
\begin{equation}
\ma\s = \mass_F\s \na\s (\mass_C\s)^{-1}
\end{equation}

\begin{rem}We also have the following classical expressions for matrices $\na\s$ and $\ma\s$:
\begin{equation}\label{eq:def_effortCF2}
\left.\begin{aligned}
\na\s &= \left(\mass_F\s \right)^{-1} \mass_{FC}\s \\
\ma\s &=  \mass_{FC}\s\left(\mass_C\s \right)^{-1} 
\end{aligned}\right\}\text{ with } \mass_{FC}\s=\int_\Gamma \shapeF\sT \shapeC\s d\Gamma
\end{equation}
When testing the balance of tractions in the interface displacements we obtain the classical preservation of the mechanical work which writes:
\begin{equation}\label{workpreserve}
\na\sT \ma\s = \matid
\end{equation}
\end{rem}

\begin{rem}
 Of course, for unchanged (non-refined) subdomains, we have  $\na\s=\ma\s=\matid$.
\end{rem}

\subsection{FETI and BDD algorithms with nonconforming interfaces}\label{ssec:alginc}
The standard BDD algorithm (see for instance Algorithm~$1$ in \cite{Rey14}) is barely modified, only the assembling operators change depending if they manipulate displacement or forces:
\begin{itemize}
\item BDD assembling operator becomes $\pa\s\na\sT$,
\item BDD scaled assembling operator becomes $\tpa\s\ma\sT$,
\end{itemize}
The final algorithm is summarized in Algorithm~\ref{alg:bdd_incomp}. 

\begin{algorithm2e}[th]\caption{Block-BDD with incompatibility, unknown $\efDep=[\efdep,\tilde{\efdep}]$}\label{alg:bdd_incomp}
$\efDep=\mathtt{Initialize}(\Force\s)$ \;  %
$(\Lam\s_D,\efDep_D\s)=\mathtt{Solve}_D(\na\s \pa\sT\efDep,\Force\s)$\;
Compute residual $\RRes=\sum_s\pa\s \na\sT \Lam\s_D$\;
Define local traction $\tLam\s=\ma\s\tpa\sT\RRes$ \tcp*[r]{$\Lam_N\s=\Lam\s_D-\tLam\s$} 
$\tdep\s=\mathtt{Solve}_N(\tLam\s,0)$ \tcp*[r]{$\efDep_N\s=\efDep_D\s-{\tDep}\s $}
Preconditioned residual $\bZ=\sum_s\tpa\s  \ma\sT \tDep\s$ \;
Search direction $\bW= \proj_1 \bZ$\;
\While{$\sqrt{\RRes^T\bZ}>\epsilon$}{%
  $(\pLam\s_D,\pDep_D\s)=\mathtt{Solve}_D(\na\s \pa\sT \bW,0)$\;
  $\bP=\sum_s\pa\s \na\sT \pLam\s_D$\;
  $\alpha=(\bP^T\bW)^{-1}(\RRes^T\bZ)$\;
  $\efDep\leftarrow \efDep+ \bW\alpha$  \tcp*[r]{$\begin{aligned}&\efDep_D\s\leftarrow \efDep_D\s+\pDep_D\s \alpha \\&\Lam_D\s\leftarrow \Lam_D\s+\pLam_D\s\alpha \end{aligned} $}
  $\RRes \leftarrow \RRes-\bP\alpha $\;
  $\tLam\s=\ma\s \tpa\sT\RRes$ \tcp*[r]{$\Lam_N\s=\Lam\s_D-\tLam\s$}
  $\tDep\s=\mathtt{Solve}_N(\tLam\s,0)$ \tcp*[r]{$\efDep_N\s=\efDep_D\s-{\tDep}\s $}
  $\bZ = \sum_s\tpa\s \ma\sT \tDep_b\s$\;
  $\bW \leftarrow \proj_1 \bZ - \bW(\bP^T\bW)^{-1}(\bP^T\proj_1\bZ)  $
}%
\end{algorithm2e}
The preservation of work~(\ref{workpreserve}) 
together with the definition of the scaling operators~(\ref{eq:efintadmiss}) ensure that the reconstructed fields satisfy the required continuity or balance on the interface. More precisely in the BDD algorithm, the displacement is built as a continuous field, whereas the nodal forces $\Lam_N\s$ are balanced on the interface because:
\begin{equation}
\begin{aligned}
\sum_s \pa\s \na\sT \Lam_N\s&= \sum_s \pa\s \na\sT  (\Lam\s_D-\tLam\s)\\
&=(\matid  -  \sum_s \pa\s \na\sT \ma\s\tpa\sT) \RRes \\
&= (\matid - \sum_s \pa\s \tpa\sT) \RRes =0
\end{aligned}
\end{equation}

\begin{rem}
Similarly, FETI algorithm can be derived from Algorithm~2 in \cite{Rey14} by simply replacing the assembling operators:
\begin{itemize}
\item FETI assembly operator becomes $\da\s\ma\sT$,
\item FETI scaled assembling operator becomes  $\tda\s\na\sT$.
\end{itemize}
\end{rem}

\begin{rem}
When the solver is converged, because of the incompatible meshes, the displacement fields $(\dep_D\s)_s$ and $(\dep_N\s)_s$ are no longer identical.
\end{rem}

\subsection{Assessment}\label{ssec:assessinc}

Let us reuse the example of previous section. Figure~(\ref{fig:error_maps_coarse}) plots the error map of the adjoint problem on a first coarse mesh, it is essentially localized within the central subdomain; where the loading of the adjoint problem is defined.

\begin{figure}[ht]\centering
\hfill\subfloat[Coarse mesh\label{fig:error_maps_coarse}]{ \includegraphics[width=.3\textwidth]{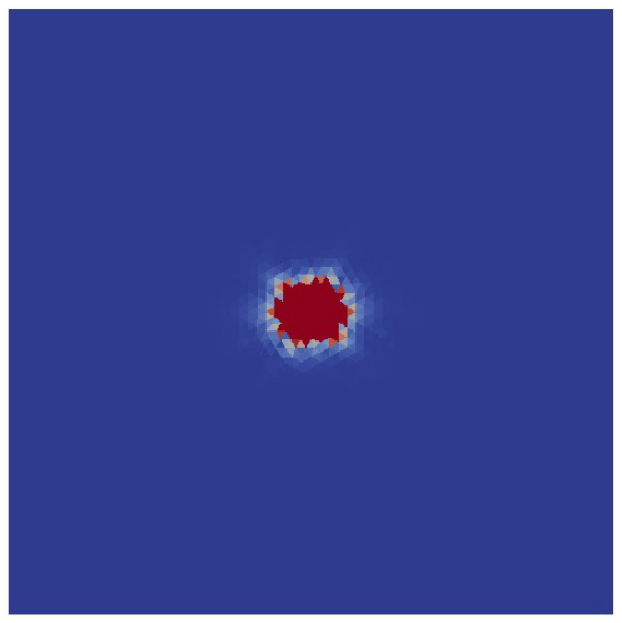}}
\hfill\subfloat[Locally refined mesh\label{fig:error_maps_coarse_fined}]{ \includegraphics[width=.3\textwidth]{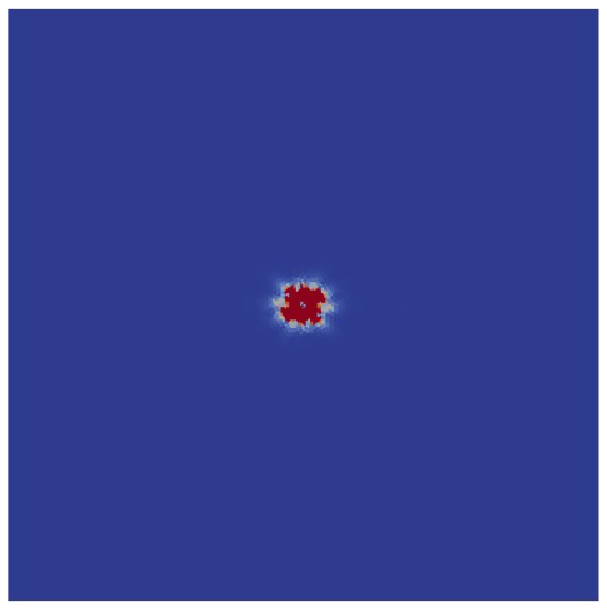}}
\hfill\subfloat[]{ \includegraphics[width=.1\textwidth]{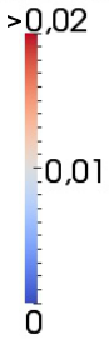}}\hfill\caption{Effect of local refinement (map of the error estimator) on the adjoint problem}\label{fig:square_adjoint}
\end{figure}
Accordingly, from the point of view of the adjoint (resp. forward) problem, we need to refine the mesh of the central subdomain. 
Three cases are considered: two with conforming meshes (i) a coarse mesh (given $H$) (ii) a fine mesh ($h=H/4$), and one with non-conforming meshes due  local refinement only in the central subdomain ($h=H/4$, see Figure \ref{fig:poutre_decoupage_and_mesh} for a schematic presentation with a coarser mesh than in actual computations). In all cases, we compute the error estimator at convergence of the solver. The number of Conjugate Gradient iterations is identical for the 3 meshes ($9$ iterations). The results are summarized in Table~\ref{tab:poutre_comp_incomp}. 
When the central subdomain is refined, the error estimator for the adjoint problem decreases while the error estimator of the forward problem remains more or less unchanged. 
As a consequence, the use of local refinement (independently by subdomains) enables us to obtain improved bounds on the quantity of interest  ($-32\%$) with a simple implementation and less degrees of freedom ($+38\%$ \#dof) than would be required by the full fine mesh ($+300\%$ \#dof for a $-65\%$ decrease of the estimator).

\begin{figure}[ht]
\centering
\includegraphics[width=.4\textwidth]{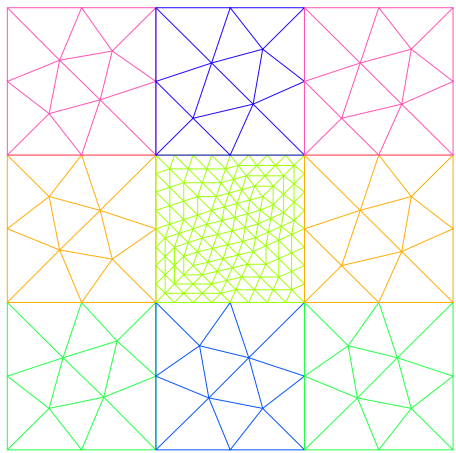}
\caption{Hierarchically refined mesh on the central subdomain}\label{fig:poutre_decoupage_and_mesh}
\end{figure}

\begin{table}[ht]\centering%
\begin{tabular}{|c|c|c|c|c|c|c|c|}
\hline 
case & $e$ & $\widetilde{e}$ &  $I_H$  & $I_{HH,2}$  &$\frac{1}{2}e\widetilde{e}$   & \#dof \\ 
\hline \ 
Coarse mesh& 2.418  & 10.23   &  2.916& 3.8575 $10^{-4}$& 12.37
&   8738 \\ 
Fine mesh & 1.216  &  7.020 &2.916  & -9.695 $10^{-4}$& 4.268
&   34988  \\ 
\hline
Refined central SD & 2.404& 7.020 & 2.9169 & 3.8575 $10^{-4}$& 8.438&   12122  \\
\hline 
\end{tabular} 
\caption{Comparaison between conforming and non conforming cases, $I_{ex}=2,916$}\label{tab:poutre_comp_incomp}
\end{table}

Figure~\ref{fig:error_maps_coarse_fined} plots the map of error for the adjoint problem when only the central subdomain is refined. The source of error due to the non-conforming interface (boundary of the central subdomain) is clearly negligible.\medskip

Of course, the use of locally refined meshes is questionable if the zone of high contribution to the error is close to a non-matching interface: in that case the driving of the fine kinematic by the coarse interface is an extra source of error. To evaluate that problem, we study how the efficiency of the local refinement is altered by the distance $d$ between the support of the quantity of interest  $\omega=[-\frac{l}{10}-d;\frac{l}{10}-d]\times[-\frac{l}{10};\frac{l}{10}]$ and the interface $\Gamma$. 

For each configuration (characterized by the distance $d$) we compute the relative discretization error (discretization error divided by the energy of the finite element solution) of the adjoint problem $\rho_{Hh}$ (resp. $\rho_{h}$) when $h=H/4$ only in the central subdomain (resp. on all subdomains). Figure~\ref{fig:etude_distance} plots the ratio $\rho_{Hh}/\rho_{h}$ as a function of the distance $d/d_0$ (where $d=d_0$ corresponds to  the quantity of interest at the center of the subdomain). 
The quality of the refined estimator is degraded only when $d/d_0<10^{-1}$ which corresponds to a quantity of interest very close to the interface. Thus the idea to conduct local refinements remains valid in a wide range of configurations.

 \begin{figure}[ht]\centering
\begin{tikzpicture}
\begin{semilogxaxis}[
width = 0.6\textwidth,
xlabel=$\frac{d}{d_0}$,
legend style={at={(0.06,0.35)}, anchor=north east}] 
\addplot [ mark=triangle*] table[x=distance2,y=quotient_rel] {cube_distance.txt};
\addlegendentry{ {\scriptsize $\frac{\rho_{Hh}}{\rho_h}$ } }
\end{semilogxaxis}
\end{tikzpicture}
\caption{Influence of the distance d}\label{fig:etude_distance}
\end{figure}
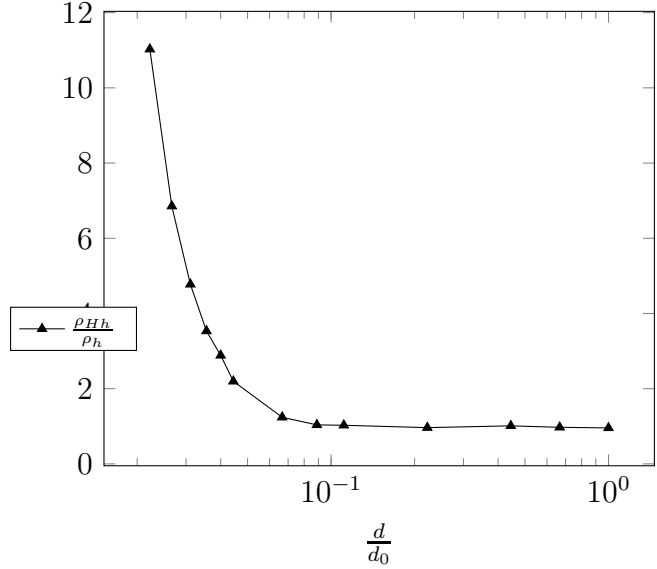

%% file: Adaptive_strategy.tex
The aim of these assessments is to show how easily the error bound can be improved by the refinement strategy described in Section~\ref{section_2}. 

We always use the stopping criterion proposed in Section~\ref{section_1}: iterations are stopped when the algebraic error becomes negligible (10 times smaller) with respect to the contribution of the discretization estimated at the first iteration, for both forward and adjoint problems. We verify that the discretization error computed with the new interface fields  barely varies, as could be expected from \cite{Rey14}. 

After a first resolution on a coarse mesh, we analyze the contributions of subdomains to the forward and adjoint error estimator using the following measures:
\begin{equation}
 \eta\s=\frac{\ecr{\dep_N\s,\hat{\sig}_N\s}{\Omega\s}}{\sqrt{\sum_s\ecrc{\dep_N\s,\hat{\sig}_N\s}{\Omega\s}}} \qquad \text{and}\qquad
\widetilde{\eta}\s=\frac{\ecr{\depadj_N\s,\hat{\sigadj}_N\s}{\Omega\s}}{\sqrt{\sum_s\ecrc{\depadj_N\s,\hat{\sigadj}_N\s}{\Omega\s}}}
\end{equation}
From that information we deduce h-refinement strategies and evaluate their ability to improve the estimator.

\subsection{The Gamma-shaped structure}
We consider a Gamma-shaped structure clamped on its base and submitted to prescribed  traction and shear on the upper-right side. The problem is linear elastic with unit Young modulus and 0.3 Poisson coefficient. The initial mesh is made out of first order triangles. The structure is decomposed into 7 domains and we use the classical BDD solver. The quantity of interest is the mean value of the strain $\varepsilon_{yy}$ over the region $\omega$ located in the seventh subdomain, as illustrated in Figure \ref{fig:GSS_adjoint}.  We build statically admissible stress fields with the Flux-free technique \cite{Cot09,Par06}. Each star-patch problem is solved with a forth order interpolation.

\begin{figure}[ht]
\centering
\includegraphics[width=.5\textwidth]{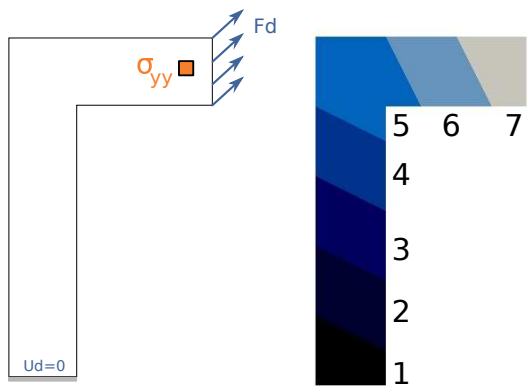}\caption{{\footnotesize Loading of forward (blue) and adjoint problems (orange)}}\label{fig:GSS_adjoint}
\end{figure}

The contributions $\eta\s$ and $\widetilde{\eta}\s$ to the error estimator on the initial mesh are presented in Figure~\ref{Fig:GSS_uniforme}. We observe that the major contribution to the error for the forward problem is located in the fifth subdomain (because of the corner). 
The objective is to refine this domain in order to obtain a level of error comparable to subdomain~1 (which is the second largest contributor). It implies to divide the error in the fifth subdomain by 2.8. The convergence in this subdomain is driven by the corner singularity. Therefore, we decide to perform three hierarchical h-refinements in the fifth subdomain. The quantity of interest being localized in Subdomain~7, we also process 3 hierarchical h-refinements in that subdomain.

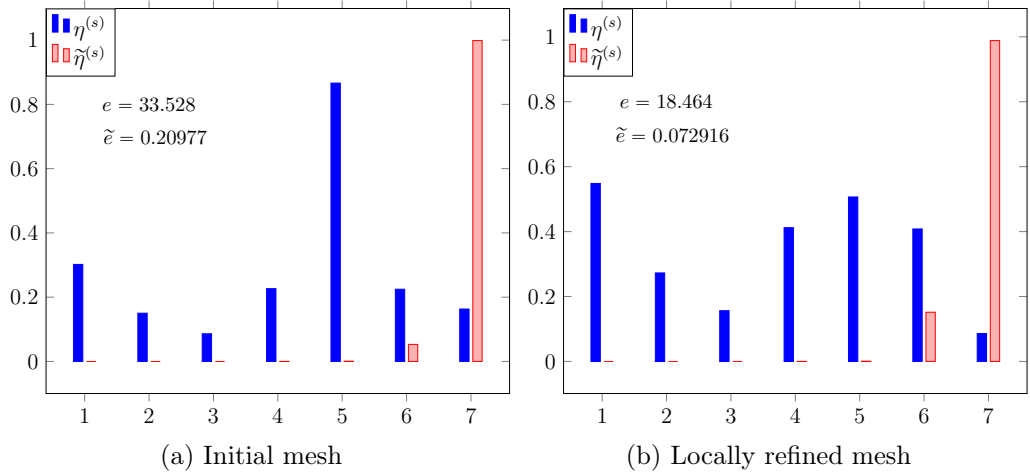
\begin{figure}
\subfloat[Initial mesh\label{Fig:GSS_uniforme}]{
\begin{tikzpicture}[scale=.7]
\begin{axis} [width= 0.7\textwidth, ybar,bar width=5pt,xtick={1,2,3,4,5,6,7},legend style={at={(0,1)},anchor=north west}]
  \addplot [fill=blue,draw=blue] table[x=sd,y=e_ref_rel]{GSS_res1.txt};
  \addplot table[x=sd,y=e_adj_rel]{GSS_res1.txt};
\legend{$\eta\s$, $\widetilde{\eta}\s$}
\node at (axis cs:2,0.8) {{\small $e=33.528$}};
\node at (axis cs:2.1,0.7) {{\small $\widetilde{e}=0.20977$}};
\end{axis}
\end{tikzpicture}}
\subfloat[Locally refined mesh\label{Fig:GSS_res3}]{
\begin{tikzpicture}[scale=.7]
\begin{axis} [width= 0.7\textwidth,ybar,bar width=5pt,xtick={1,2,3,4,5,6,7},legend style={at={(0,1)},anchor=north west}]
  \addplot [fill=blue,draw=blue] table[x=sd,y=e_ref_rel]{GSS_res3.txt};
\addplot table[x=sd,y=e_adj_rel]{GSS_res3.txt};
\node at (axis cs:2,0.8) {{\small $e=18.464$}};
\node at (axis cs:2.1,0.7) {{\small $\widetilde{e}=0.072916$}};
\legend{$\eta\s$, $\widetilde{\eta}\s$}
\end{axis}
\end{tikzpicture}}\caption{Distribution of the error estimator within subdomains}\label{Fig:GSS_reszozo}
\end{figure}

As expected, we can see on Figure \ref{Fig:GSS_res3} that the forward error in Subdomain~5 and Subdomain~1 are nearly the same on the new mesh.  We also observe that the error of the adjoint problem has decreased (due to the refinement of Subdomain~7). 
In Table~\ref{tab:GSS_global}, we sum up the global values of discretization errors, the quantities~$I_H$ and~$I_{HH,2}$. 
\begin{table}[ht]
 \centering 
\begin{tabular}{|c||c|c|c|c|c|c|}
\hline 
Mesh & $e$ & $\widetilde{e}$ &  $I_H$  & $I_{HH,2}$  & $\frac{1}{2}e\widetilde{e}$    \\ 
\hline \ 
Uniform &33.528 &  0.20977 & -0.39157 &-0.0010045  & 3.5165\\ 
\hline
Locally refined &18.464 &0.072916  &-0.39006  &  -0.00050892& 0.67316 \\
\hline 
\end{tabular} \caption{Performance of local refinement on the $\Gamma$-shaped structure}\label{tab:GSS_global}
\end{table}

\subsection{The pre-cracked structure}

We now consider the pre-cracked structure decomposed into 16 subdomains as illustrated in Figure~\ref{fig:Diez_adjoint}. The displacements at the base of the structure and on the bigger hole are imposed to be zero. The upper-left part and the second hole are subjected to a constant unit pressure. The quantity of interest is the mean of the stress component $\sigma_{xx}$  on a region $\omega$ close to the crack. In Figure \ref{fig:Diez_adjoint}, the loading of the reference problem is in blue and the loading of adjoint problem is in orange. We used the FETI algorithm (dual approach) to solve the interface problem and the statically admissible stress fields are built using the Flux-free technique \cite{Cot09,Par06} with forth order interpolation for problems on star-patches.

\begin{figure}[ht]
\centering
\includegraphics[width=.5\textwidth]{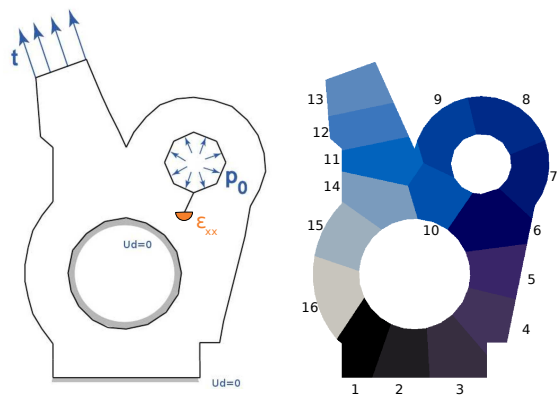}\caption{{\footnotesize Loading of forward (blue) and adjoint problems (orange), domain decomposition}}\label{fig:Diez_adjoint}
\end{figure}

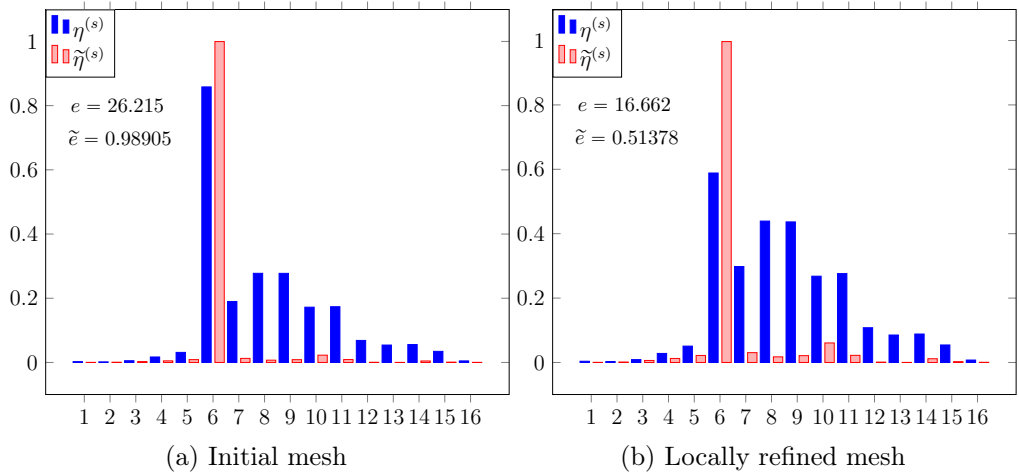
\begin{figure}[ht]
\subfloat[Initial mesh\label{fig:Diez_uniforme}]{\begin{tikzpicture}[scale=.7]
\begin{axis} [width= 0.7\textwidth,ybar,bar width=5pt,xtick={1,2,3,4,5,6,7,8,9,10,11,12,13,14,15,16},legend style={at={(0,1)},anchor=north west}]
  \addplot [fill=blue,draw=blue]  table[x=sd,y=e_ref_rel]{diez_res1.txt};
\addplot table[x=sd,y=e_adj_rel]{diez_res1.txt};
\legend{$\eta\s$, $\widetilde{\eta}\s$}
\node at (axis cs:2.3,0.8) {{\small $e=26.215$}};
\node at (axis cs:2.4,0.7) {{\small $\widetilde{e}=0.98905$}};
\end{axis}
\end{tikzpicture}}
\subfloat[Locally refined mesh\label{fig:Diez_res4}]{\begin{tikzpicture}[scale=.7]
\begin{axis} [width= 0.7\textwidth,ybar,bar width=5pt,xtick={1,2,3,4,5,6,7,8,9,10,11,12,13,14,15,16},legend style={at={(0,1)},anchor=north west}]
  \addplot [fill=blue,draw=blue] table[x=sd,y=e_ref_rel]{diez_res4.txt};
\addplot table[x=sd,y=e_adj_rel]{diez_res4.txt};
\legend{$\eta\s$, $\widetilde{\eta}\s$}
\node at (axis cs:2.3,0.8) {{\small $e=16.662$}};
\node at (axis cs:2.4,0.7) {{\small $\widetilde{e}=0.51378$}};
\end{axis}
\end{tikzpicture}}
\caption{Distribution of error within subdomains}\label{fig:Diez_zozo}
\end{figure}
In Figure \ref{fig:Diez_uniforme}, we observe that the discretization errors of the forward and adjoint problem are quasi totally located in the sixth subdomain.  Therefore, we perform hierarchical refinement in this subdomain in order to reach a level of forward error comparable to the second larger contributor (Subdomain~8). That corresponds to dividing  the error in Subdomain~6 by a factor 3. Because of the singularity in the subdomain, convergence is expected to evolve like $\sqrt{h}$ so we perform~4 hierarchical refinements (each coarse element edge is divided into~16 fine edges).

As can be seen in Figure~\ref{fig:Diez_res4}, after local refinement the contributions of Subdomains~6 and~8 to the forward error are equivalent, and the adjoint error also decreased, as expected. The global performance is summed up  in Table~\ref{tab:Diez_global}. In the end, the bound on the estimator on the quantity of interest was decreased by a factor 3.

 \begin{table}[ht]\centering%
\begin{tabular}{|c||c|c|c|c|c|c|}
\hline 
Mesh & $e$ & $\widetilde{e}$ &  $I_H$  & $I_{HH,2}$  & $\frac{1}{2}e\widetilde{e}$     \\ 
\hline \ 
Uniform & 26.215 & 0.98905  &  2.4915& 3.1935 &12.964 \\ 
\hline
Locally refined & 16.662& 0.51378 & 3.2165 & 0.086055 &4.2803 \\
\hline  
\end{tabular} \caption{Performance of local refined for the cracked structure}\label{tab:Diez_global}
\end{table}

%% file: Conclusion.tex
In this article, we extended two classical results of goal-oriented error estimation to problems solved by domain decomposition methods. We solved simultaneously both forward and adjoint problems on the same mesh using a block conjugate gradient algorithm and proposed a fully parallel procedure to recover admissible fields and estimate the error on the quantity of interest. Moreover, the error bounds separated the contribution of the discretization from the contribution of the iterative solver which lead to defining of a stopping criterion of the iterative solver that prevented oversolving. 

We also showed that the method could be easily extended to the case of locally enriched kinematics (using $h$ or $p$ refinements or dedicated enrichments) only constrained to contain the original coarse kinematic on the boundary of subdomains. This property provided a simple parallel way to improve estimators as soon as the singularities (local quantity of interests, crack tips, etc) were not too close to the interface. The concepts were validated on academic examples using the most basic $h$-refinement strategy. Of course using ad-hoc enrichment instead of brutal $h$-refinement should improve the performance of the method. 

Future work will focus on exploiting the fact that the shape of the interface is preserved by the enrichment to define acceleration strategies, and on defining original strategies to balance the computational load between subdomains for the resolution of the refined problem.